\documentclass[10pt,conference]{ieeeconf}      
\IEEEoverridecommandlockouts
\usepackage{cite}
\usepackage{graphicx}
\usepackage{epstopdf}
\usepackage{caption}
\usepackage{subcaption}
\usepackage{amsthm}
\theoremstyle{plain}

\usepackage{amssymb} 
\usepackage{amsmath} 
\usepackage{mathtools}
\usepackage{color}
\usepackage{todonotes}
\usepackage{graphicx}
\usepackage{multirow}   
\usepackage{hhline}
\usepackage{caption} 
\usepackage{algorithm}  
\usepackage{algpseudocode}
\usepackage{mathrsfs}
\usepackage{soul} 
\makeatother

\newtheorem{lem}{Lemma}

\newtheorem{prop}{Proposition}
\newtheorem{prob}{Problem}
\newtheorem{thm}{Theorem}

\usepackage{xspace}
\usepackage{amsmath,amssymb,amsfonts}

\newcommand{\bs}[1]{\ensuremath{\boldsymbol{#1}}}

\newcommand{\bv}{\ensuremath{\bs v}\xspace}
\newcommand{\bw}{\ensuremath{\bs w}\xspace}
\newcommand{\bx}{\ensuremath{\bs x}\xspace}

\newcommand{\bX}{\ensuremath{\bs X}\xspace}

\newcommand{\bxi}{\ensuremath{\bs \xi}\xspace}

\newcommand{\pci}[1]{\ensuremath{\bs{p}_{\boldsymbol{c}_{i}}}\xspace}

\def\QED{\ensuremath{\rule[0pt]{1.5ex}{1.5ex}}\xspace}

\newcommand{\tdlateroptions}{[color=blue!20]}
\newcommand{\mytodo}{\expandafter\todo\tdlateroptions}

\newcommand{\Prob}{\mathbb{P}}
\newcommand{\Exp}{\mathbb{E}}
\newcommand{\FSRset}{\mathrm{FSReach}}
\newcommand{\FSRseta}{\mathrm{FSReach'}}
\newcommand{\probref}[1]{\emph{#1}}

\begin{document}
\title{Computation of forward stochastic reach sets: Application to stochastic, dynamic obstacle avoidance}

\author{Baisravan HomChaudhuri, Abraham P. Vinod, and Meeko M. K. Oishi
\thanks{This material is based upon work supported by the National Science Foundation under Grant Number IIS-1528047, CMMI-1254990 (Oishi, CAREER), and CNS-1329878. Any opinions, findings, and conclusions or recommendations expressed in this material are those of the authors and do not necessarily reflect the views of the National Science Foundation. 
\newline
    B. HomChaudhuri, A. P. Vinod, and M. Oishi are with 
    Electrical and Computer Engineering, University of New Mexico, Albuquerque,
NM 87131 USA.  Corresponding author email: \texttt{oishi@unm.edu}}
}
\maketitle

\begin{abstract}
We propose a method to efficiently compute the forward stochastic reach (FSR) set and its probability measure for nonlinear systems with an affine disturbance input, that is stochastic and bounded.  This method is applicable to systems with an {\em a priori} known controller, or to uncontrolled systems, and often arises in problems in obstacle avoidance in mobile robotics.  When used as a constraint in finite horizon controller synthesis, the FSR set and its probability measure facilitate {\em probabilistic} collision avoidance, in contrast to methods which presume the obstacles act in a worst-case fashion, and generate hard constraints that cannot be violated.    
We tailor our approach to accommodate rigid body constraints, and show convexity is assured so long as the rigid body shape of each obstacle is also convex.  We extend methods for multi-obstacle avoidance through mixed integer linear programming (with linear robot and obstacle dynamics) to accommodate chance constraints that represent the FSR set probability measure.  We demonstrate our method on a rigid-body obstacle avoidance scenario, in which a receding horizon controller is designed to avoid several stochastically moving obstacles while reaching a desired goal.  Our approach can provide solutions when approaches that presume a worst-case action from the obstacle fail.
\end{abstract}


\section{Introduction}
\label{sec:introduction}
Navigation in stochastic, dynamic environments is a challenging task in a variety of application domains, including robotics, autonomous driving, unmanned aerial vehicles, and other transportation systems.  In any realistic environment, reliable, collision-free navigation is paramount, and must be implementable in a manner that is amenable to real-time operation.  For an environment with stochastic, dynamic obstacles, accurate prediction of potential obstacle locations, as well as likelihood of obstacle occupancy at those locations, constrains navigation.  Further, physical constraints arising due to, e.g., separation constraints or the geometry of rigid body (not point-mass) obstacles must also be incorporated. Synthesizing these constraints into existing frameworks for robot navigation requires efficient representation of obstacle avoidance constraints.   
We propose a method to compute the forward stochastic reachable (FSR) set and its probability measure for dynamical systems with affine disturbance input, motivated by the problem of collision-free navigation in an environment with many stochastic, dynamic, rigid-body obstacles.

A variety of approaches have been proposed for navigation amidst dynamic obstacles.  Some formulations are reactive, meaning that instead of incorporating predictions of the obstacle location, they take action according to the current measurement only \cite{rimon1992exact}.  Predictive formulations, in contrast, anticipate future motion, sometimes through the use of a constrained finite-horizon optimization framework, with constraints arising due to robot dynamics and predictions of obstacle position.  These methods involve solving a mixed-integer linear program \cite{schouwenaars2002safe,schouwenaars2001mixed}, a mixed-integer quadratic program, \cite{mellinger2012mixed}, or using sampling based methods \cite{karaman2011sampling,ses}.  

Predictions of obstacle location are dependent upon assumptions about obstacle dynamics and stochastic properties.  
For non-holonomic point-mass obstacles, velocity obstacles \cite{fiorini1998motion} exploit a closed-form solution to approximate the forward reachable set over a finite horizon, presuming a constant velocity.  For probabilistic obstacles with bounded uncertainty, 
a variety of approaches compute the set of {\em all} possible obstacle states, 
but not the likelihood of obstacle occupancy, 
with application to robotics \cite{wu2012guaranteed,chung2006coordinated,althoff2014online,lin2015collision}, and to  automotive vehicles \cite{althoff2014online}. 
These approaches are conservative, in that they rule out potentially large areas of the state-space, even if obstacle occupancy is low. 
Some non-conservative solutions involve receding horizon controllers to avoid collision at the expected future location of the obstacles \cite{du2012robot}, but 
can still lead to collision with excessively high variance or with multiple obstacles.

%

Strict assurances of safety are possible with 
backward reachable sets \cite{Mitchell2007}.  
A controller is constructed by solving the Hamilton-Jacobi-Issacs equation \cite{mitchell2005time, fisac2015reach, ding2010robust, takei2012time} presuming the worst case realization of the obstacle or disturbance (also referred to as the `min-max' or {\em robust} solution). 
Methods based on the backwards reachable set often suffer from computational complexity that is exponential in the dimensionality of the state space.  Low-dimensional systems in 
aerospace and automotive applications have been explored \cite{mitchell2005time,Marg1,Ding2011,Gillula2011,chen2015safe}. 
However, for stochastic obstacles with bounded input, these methods can be overly conservative, especially when the disturbance variance is large, or in scenarios with multiple moving obstacles, when the collision-free space diminishes quickly as the time horizon increases. 
We previously used backwards reachable sets to weight probabilistic roadmaps \cite{malone2014stochastic} and artificial potential fields \cite{chiang2015aggressive}, but without assurances of safety, since the sets could only be computed pairwise between the robot and a single obstacle, due to computational complexity.

Accurate predictions are particularly key for dynamic, stochastic obstacles. 
Probabilistically safe trajectories \cite{aoude2013probabilistically,ses,best2015bayesian,havlak2014discrete} exploit knowledge of the likelihood of obstacle location.
Predictions have been accomplished via Monte Carlo simulations  \cite{aoude2013probabilistically,ses,best2015bayesian} and via Gaussian mixture models \cite{havlak2014discrete}.  While these methods have an appealing flexibility and simplicity, the quality of the prediction of obstacle location is highly dependent on the number of particles used, and it is in general difficult to estimate a priori the number of particles required for a desired quality.


We propose an alternative method of prediction, based on forward stochastic reachable sets, that uses not only the set of states that the obstacle can reach, but also the likelihood of occupancy of all possible obstacle locations.  We present an iterative formula for the computation of the FSR set for nonlinear dynamical systems with affine disturbance input, which is exact for a bounded, countable disturbance set, and a method   
for the computation of the FSR probability measure.  
We extend this approach to rigid-body obstacles with convex geometry through the use of an indicator function that represents the body geometry.  We derive an {\em occupancy function} for a rigid-body obstacle that can be used to generate an exact set of states that the robot should avoid to avoid collision with at least a certain likelihood.  Superlevel sets of the occupancy function become inequality constraints for integer programming based methods for obstacle avoidance \cite{schouwenaars2002safe,schouwenaars2001mixed,mellinger2012mixed}.
For scenarios with multiple obstacles, we use an over-approximation which can be expressed as the union of convex sets, since the union of superlevel sets of occupancy functions for each obstacle is not necessarily convex.  Our results indicate that our method provides feasible solutions when robust methods that exploit a min-max approach fail.

The main contributions of this paper are: 1) a method to efficiently compute the forward stochastic reachable set and  probability measure for systems with bounded, affine, stochastic disturbance, and 2) formulation of occupancy constraints, based on the FSR probability measure, as the union of convex sets, 
to generate probabilistic safe robot trajectories in presence of multiple stochastic, dynamic, rigid-body obstacles, using existing integer programming-based collision avoidance methods.

The paper is organized as follows: Section \ref{sub:problem_formulation} describes the problem formulation and mathematical preliminaries.  Section~\ref{sec:computeReach} formulates the forward stochastic reachability iteration for nonlinear as well as linear systems.  We apply our methods to the rigid-body obstacle avoidance problem in Section \ref{sec:motionPlanning}, and provide
conclusions and directions for future work in Section~\ref{sec:conc}.


\section{Preliminaries and Problem formulation}
\label{sub:problem_formulation}

Consider the discrete-time time-invariant dynamical system,
\begin{align}
    \bx[t+1]&=f(\bx[t])+g(\bw[t])\label{eq:sys_orig}
\end{align}
with state $\bx[t]\in \mathcal{X}\subseteq \mathbb{R}^n$, disturbance $w[t]\in
\mathcal{W}\subseteq \mathbb{R}^p$, and Borel-measurable functions $f:
\mathbb{R}^n \rightarrow \mathbb{R}^n$ and $g: \mathbb{R}^p \rightarrow
\mathbb{R}^n$. We define the inital set $\mathcal I$, and an initial condition $\bx[0]\in \mathcal{I}\subseteq
\mathcal{X}$. The disturbance
set $\mathcal{W}$ is bounded and countable, and the
random vector $\bw[t]$ is defined in a probability space $(\mathcal{W},\sigma(
\mathcal{W}), \Prob_{\bw})$. We assume the random vector has a known probability
mass function. For a countable sample space, the probability measure
$\Prob_{ \bw}$ defines a probability mass function $\psi_{\bw}[\cdot]:
\mathbb{R}^p \rightarrow [0,1]$ such that for $\bar{z}=(z_1,z_2,\ldots,z_p)\in
\mathbb{R}^{p}$, $\Prob_{\bw}\{\bw=\bar{z}\}=\psi_{\bw}[\bar{z}]$. We define an indicator function $\mathbf{1}_{Y}(\bar{y}): \mathbb R^n \rightarrow \{0, 1\}$ such that it takes on the value 1 for $\bar{y} \in Y$ and 0 otherwise.  We use $| \cdot |$ to indicate cardinality. The $p \times p$ identity matrix is denoted $I_p$.

The dynamics \eqref{eq:sys_orig} are quite
general, and include affine noise perturbed systems with known state-feedback
based inputs or open-loop controllers.  We assume that the disturbance
process $\bw[t]$ is an i.i.d. random process with respect to
time.  For a known initial condition, the state $\bx[t+1]$ is a random vector
due to $\bw[t]$. A random initial condition  $\bx[0]$ is defined in a probability space $(\mathcal{I},\sigma(\mathcal{I}), \Prob_{\bx[0]})$ with probability measure $\Prob_{\bx[0]}\{\bx[0]=\bar{z}\}= \psi_{\bx[0]}[\bar{z}]$. 



By defining a random vector $\bv[t]=g(\bw[t])$ in the probability space $(
\mathcal{V}, \sigma( \mathcal{V}) , \Prob_{ \bv})$, 
$\eqref{eq:sys_orig}$ can be simplified to
\begin{align}
    \bx[t+1]=f(\bx[t])+\bv[t].\label{eq:sys}
\end{align}
Given an initial condition $\bar{x}_0$ and a
sequence of random vectors $\{\bv[t]\}_{t=0}^{t=\tau}$, the trajectory of
$\eqref{eq:sys}$ is completely characterized by a random process defined
as $\bx[\tau]=\bxi(\tau;\bar{x}_0):[0,T] \rightarrow \mathcal{X}$. Therefore, the random
vector $\bx[\tau]$ is defined in the probability space $( \mathcal{X},\sigma(
\mathcal{X}), \Prob^{\tau,\bar{x}_0}_{\bx})$. Here, $\Prob^{\tau,\bar{x}_0}_{\bx}$ is
induced from the product measure of $\Prob_{ \bv}$ 
since $\bv[t]$ is an i.i.d random process.  When $f(\cdot)$ and $g(\cdot)$ are
linear transformations $A\in \mathbb{R}^{n\times n}$ and $B\in
\mathbb{R}^{n\times p}$, respectively, we have a linear time-invariant
system
\begin{align}
    \bx[t+1]&=A\bx[t]+B\bw[t],\label{eq:sys_linear}
\end{align}
and with $\bv[t]=B\bw[t]$, this becomes
\begin{align}
    \bx[t+1]&=A\bx[t]+\bv[t].\label{eq:sys_linear_v}
\end{align}
We are interested in determining those states that can be reached with non-zero probability, as well as the likelihood of reaching those states.

For the discrete-time systems defined in \eqref{eq:sys} and
\eqref{eq:sys_linear_v}, we define the forward stochastic reach set as
\begin{align}
    \FSRset(\tau, \mathcal{I})&=\big\{\bar{y}\in
    \mathcal{X}\vert\exists\bar{x}_0\in \mathcal{I},{\{\bar{z}[t]\}}_{t=0}^{t=\tau}\mbox{ with }
    \nonumber \\
\Prob_{ \bv}\{\bv[t]=\bar{z}&[t]\}>0\ \forall t\in [0,\tau]\mbox{ s.t.
}\bxi(\tau;\bar{x}_0)=\bar{y}\big\}.\label{eq:FSRreach}
\end{align}
Here, ${\{\bar{z}[t]\}}_{t=0}^{t=\tau}$ is a realization of the random process
${\{\bv[t]\}}_{t=0}^{t=\tau}$ that can occur with non-zero probability and
$\bar{z}[\cdot]\in \mathbb{R}^n$.  We define the forward stochastic reach
probability measure (FSRPM) at time $t$ as the probability measure associated
with the state at time $t$, $\Prob^{t}_{\bx}$. For a countable disturbance set
$\mathcal{V}$ and $\bar{y}\in \mathbb{R}^n$, the FSRPM is defined by 
\begin{align}
    \Prob^{t}_{\bx}\{\bx[t]=\bar{y}\}=\sum_{\bar{z}\in
    \mathcal{I}}\Prob^{t,\bar{z}}_{\bx}\{\bx[t]=\bar{y}\}\Prob_{\bx[0]}\{\bx[0]=\bar{z}\}.\label{eq:Probtx}
\end{align}
The existence of a probability mass function for the
FSRPM \eqref{eq:Probtx} is guaranteed, since
the Borel-measurable functions in \eqref{eq:sys} preserve measurability.

\begin{lem}\label{lem:FSRset}
For a countable set $ \mathcal{V}$, $\FSRset(t, \mathcal{I})=\{\bar{y}\in
\mathcal{X}:\psi_{\bx}[\bar{y};t]> 0\}$.
\end{lem}
Lemma \ref{lem:FSRset} arises by construction, and asserts that the forward stochastic
reach set \eqref{eq:FSRreach} is the support of the corresponding
FSRPM \eqref{eq:Probtx}.
Note that the equality in Lemma~\ref{lem:FSRset} would be
\emph{almost sure} if the additional restriction of $ \Prob_{
\bv}\{\bv[t]=\bar{z}[t]\}>0\ \forall t\in [0,\tau]$ were not imposed in
\eqref{eq:FSRreach}.

\begin{prob}
   Given the affine, stochastic dynamics (\ref{eq:sys}), initial condition $\bx[0]\in\mathcal{I}$ and its distribution $\psi_{\bx[0]}$, disturbance set $\mathcal{V}$, disturbance probability mass function $\psi_{\bv}[\cdot]$, compute the forward stochastic reach set $\FSRset(t, \mathcal{I})$
   and the forward stochastic reach probability measure $\psi_{\bx}[\cdot;t]$ at time $t$, in an iterative fashion.  
   \label{prob:basic}
\end{prob}

We are motivated by problems in dynamic, stochastic obstacle avoidance.  Specifically, we wish to describe those states which are associated with a likelihood of collision with a rigid-body obstacle that is at or above a level $\alpha \in [0,1]$.  For a single obstacle scenario, this is the $\alpha$-superlevel set of the obstacle's {\em occupancy function}, to be defined precisely later.  
We require a computationally tractable formulation of the superlevel set of the occupancy function, that enables the use of integer programming based methods for obstacle avoidance. We seek to then generalize this method to handle multiple dynamic, stochastic obstacles, as well.


\begin{prob}\label{prob:prob}
Construct a computationally tractable formulation of the superlevel set of the occupancy function for a rigid-body obstacle with stochastic dynamics and convex geometry, and known initial position.  That is, represent the $\alpha$-superlevel set of the occupancy function, or an overapproximation of the $\alpha$-superlevel set of the occupancy function, as a union of convex sets at each instant $t \in [0, T]$.  
\end{prob}

\begin{prob}\label{prob:prob3}
Reconsider Problem \ref{prob:prob} for multiple rigid-body obstacles with convex geometry, and construct an overapproximation of the $\alpha$-superlevel set of the joint occupancy function that is a union of convex sets for each obstacle.
\end{prob}

\section{Forward stochastic reachability analysis}
\label{sec:computeReach}


\subsection{Nonlinear, affine dynamical system}
\label{sub:nonlin}

We assume, without loss of generality, that the empty set is the only member of the
sigma-algebra $\sigma( \mathcal{W})$ of the disturbance random vector $\bw[t]$
to have a zero probability of occurrence according to the probability measure $
\Prob_{\bw}$.

We additionally note that for 
random vectors $\bw_1,\bw_1\in \mathbb{R}^{n}$ with probability densities $\psi_{\bw_1}$, $\psi_{\bw_2}$, respectively,
\begin{enumerate}
    \item[P1)]  
    If $\bw=\bw_1+\bw_2$, then $\psi_{\bw}=\psi_{\bw_1}\ast\psi_{\bw_2}$, in which $\ast$ denotes the convolution
operation.
\item[P2)] 
   If $\bw_1$ and $\bw_2$ are independent vectors, then
        $\bw=(\bw_1,\bw_2)$ has probability density
        $\psi_{\bw}=\psi_{\bw_1}\psi_{\bw_2}$.
\end{enumerate}

The following theorem characterizes the FSR and the FSRPM using two recursive relations.
\begin{thm}\label{thm:FSR_nonlinear}
   Given the dynamics \eqref{eq:sys}, an initial condition set $ \mathcal{I}$,
   a probability mass function $\psi_{\bx[0]}[\cdot]$ over $ \mathcal{I}$, and a
   countable disturbance set $ \mathcal{V}=g( \mathcal{W})$, for every $t\in[0,T-1]$,
   \begin{align}
       \FSRset(t+1, \mathcal{I})&=f\Big(\FSRset(t, \mathcal{I})\Big)\oplus
       \mathcal{V}\label{eq:FSRset_recursive} \\
       \psi_{\bx}[\bar{y};t+1]&=\Big(\psi_{f(\bx)}[\cdot;t]\ast\psi_{\bv}[\cdot]\Big)[\bar{y}]\label{eq:FSRPM_recursive}
   \end{align}
\end{thm}
\begin{proof}
   Equation \eqref{eq:FSRset_recursive} follows from \eqref{eq:FSRreach} and the
   assumption that all non-empty members of $\sigma(\mathcal{V})$ have non-zero
   probability of occurrence. 
   Equation \eqref{eq:FSRPM_recursive} follows from the observation that
   $f(\bx[t])$ is a random vector for all $t\in[0,T]$ and Property P1.
\end{proof}

Note that \eqref{eq:FSRset_recursive} is identical to the propagation of
reachable sets as in ~\cite{kvasnica2015reachability}.
When $ \mathcal{V}$ is bounded, the forward stochastic reach sets can be
computed using existing tools for reachability, such as the
multi-parametric toolbox (MPT)~\cite{MPT3} and ellipsoidal toolbox (ET)~\cite{ellipsoid}.
We can also use Lemma~\ref{lem:FSRset} to compute these sets from their
corresponding probability measures in \eqref{eq:FSRPM_recursive}.

For the probability measure, from the definition of convolution and assumption of a countable disturbance set
$ \mathcal{V}$, we expand \eqref{eq:FSRPM_recursive} as
\begin{align}
    \psi_{\bx}[\bar{y};t]&=\sum_{\bar{z}\in
\mathcal{V}}\psi_{f(\bx)}[\bar{y}-\bar{z};t]\psi_{\bv}[\bar{z}]. \label{eq:FSRPM_recursive_sum}
\end{align}
with
\begin{align}
    \psi_{f(\bx)}[\bar{y}-\bar{z};t]&=\sum_{\bar{x}\in\FSRseta}
    \psi_{\bx}[\bar{x};t] \label{eq:FSRPM_psif}
\end{align}
and $\FSRseta=\{\bar{x}\in\FSRset(t, \mathcal{I})\vert
f(\bar{x})=\bar{y}-\bar{z}\}$.  Compared to \eqref{eq:Probtx}, equations
\eqref{eq:FSRPM_recursive_sum} and \eqref{eq:FSRPM_psif} provide a recursive
relation for $\psi_{\bx}[\cdot;t]$. For a countable 
disturbance set, \eqref{eq:FSRPM_recursive_sum} provides the exact FSR set and its probability measure.  



We summarize our solution to Problem~\ref{prob:basic} in
Algorithm~\ref{algo:FSR_thm1} for the system \eqref{eq:sys} with a countable
disturbance set $ \mathcal{V}$.
\begin{algorithm}[h]
    \caption{Forward stochastic reachable set and probability measure (Problem~\ref{prob:basic})
        for a system \eqref{eq:sys} with a countable disturbance set $\mathcal{V}$}\label{algo:FSR_thm1}
    \begin{algorithmic}[1]
        \Require Dynamics $f$, Initial set $\mathcal{I}$, Initial probability
        mass function $\psi_{\bx[0]}[\cdot]$, disturbance set $ \mathcal{V}$, disturbance probability mass function $\psi_{\bv}[\cdot]$, time instant of interest $\tau$
        \Ensure Forward stochastic reach set $\FSRset(\tau, \mathcal{I})$ and
        probability measure $\psi_{\bx}[\cdot;\tau]$
            \State $t\gets 1$ 
            \State $\FSRset(0, \mathcal{I}) \gets \mathcal{I}$
            \State $\psi_{\bx}[\cdot;0] \gets \psi_{\bx[0]}(\cdot)$
            \While{$t \leq \tau$}
                \State $\psi_{\bx}[\cdot;t] \gets 0$\Comment{Initialize FSRPM to zero}
                \ForAll{$\bar{x} \in \FSRset(t-1, \mathcal{I}),\ \bar{z} \in \mathcal{V}$} 
                    \State $\psi_{\bx}[f(\bar{x})+\bar{z};t] \gets
        \psi_{\bx}[f(\bar{x})+\bar{z};t]+\psi_{\bx}[\bar{x};t-1]\psi_{\bv}[\bar{z}]$
        \Comment{ Equations \eqref{eq:FSRPM_recursive_sum} and
        \eqref{eq:FSRPM_psif}}
                \EndFor
                \State $\FSRset(t, \mathcal{I}) \gets f(\FSRset(t-1,
                \mathcal{I}))\oplus \mathcal{V}$
                \State ${t\gets t+1}$ \Comment{Update iteration variable}
            \EndWhile
  \end{algorithmic}
\end{algorithm}
\subsection{Comparisons to the dynamic programming approach}
\label{sub:DP}

\begin{figure*}[h!]
    \centering
    \newcommand{\figw}{0.3}
    \begin{subfigure}[t]{\figw\linewidth}
        \centering
        \includegraphics[width=0.92\linewidth]{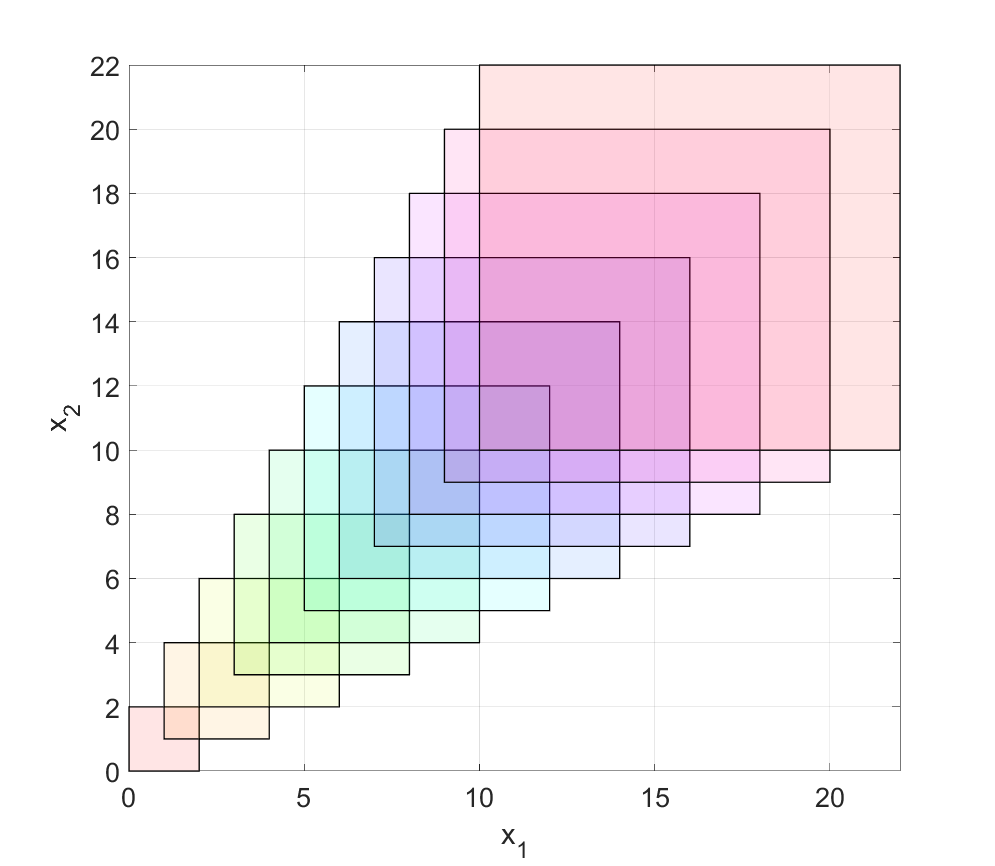}
        \caption{}
        \label{fig:FSReachSetPM}
    \end{subfigure}%
    ~ 
    \begin{subfigure}[t]{\figw\linewidth}
        \centering
        \includegraphics[width=1\linewidth]{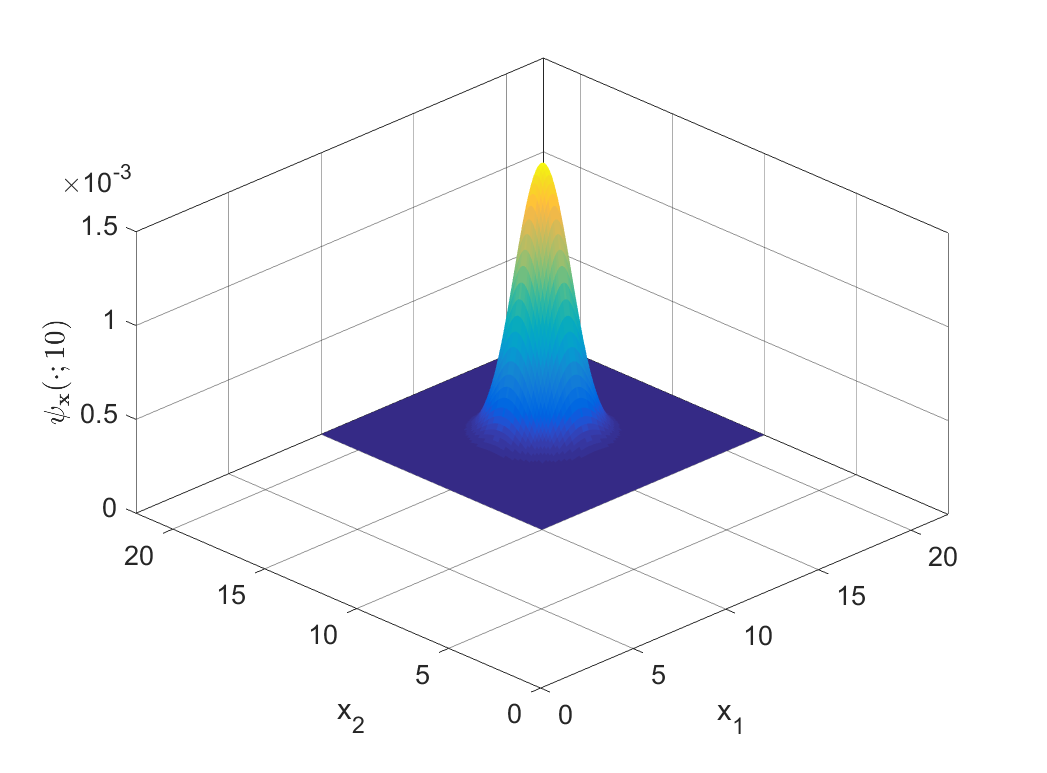}
        \caption{}
        \label{fig:FSRPMPM}
    \end{subfigure}%
    \begin{subfigure}[t]{\figw\linewidth}
        \centering
        \includegraphics[width=0.92\linewidth]{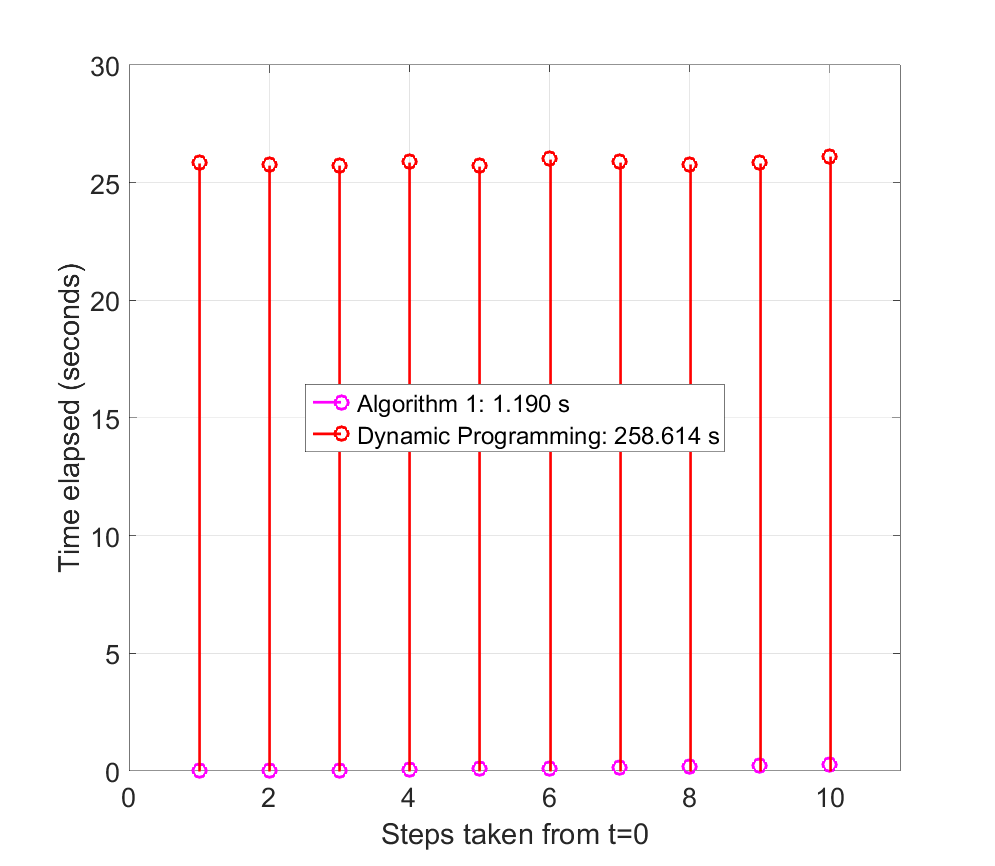}
        \caption{}
        \label{fig:FSRCompute}
    \end{subfigure}%
\caption{(a) Forward stochastic reach sets over time,  (b) forward stochastic
    reach probability measure at $t=10$, (c) comparison of run times for
Algorithm~\ref{algo:FSR_thm1} and the dynamic programming approach. The analysis
was done for the system given in
\eqref{eq:sys_PM_example}.\label{fig:FSR_example_PM}}
\end{figure*}

The dynamic programming formulation provided in~\cite{summers_verification_2010}
computes the backward stochastic reachable set for control objectives
involving safety.  It allows for calculation of either stochastic reachable or 
stochastic viable sets, and simultaneously constructs an optimal control input. 
While Problem~\ref{prob:basic} can be posed as a backward reach
problem when the dynamics are reversed in time (subject to the existence of the backward
dynamics, when 1) $f(\cdot)$ is invertible, and 2) $\bv[t]$ is
an i.i.d process)~\cite{Mitchell2007}, 
the solution to
Problem~\ref{prob:basic} does not require computation of the optimal control
action, significantly simplifying calculation. 
%
Algorithm~\ref{algo:FSR_thm1} also iterates over only those states for which the FSRPM is
positive.  Therefore, at every time instant $t$, it propagates the dynamics over a
smaller region of the state space $\FSRset(t, \mathcal{I})$ as compared to
dynamic programming.

To demonstrate, 
consider a point mass dynamics discretized in time with velocities drawn from a
truncated Gaussian distribution, 
\begin{align}
    \bx[t+1]&=\bx[t]+B_\mathrm{ex,PM}\bw[t] \label{eq:sys_PM_example} \\
    \bw[t] &\sim \mathcal{N}_\mathrm{truncated,
\mathcal{W}}(\bar{\mu},\Sigma)  \nonumber
\end{align}
with state $\bx[t]\in \mathbb{R}^2$, disturbance $\bw[t]$ is a random vector
taking values in $[1,2]^2$ following a truncated Gaussian density with mean
$\bar{\mu}=[1.5\ 1.5]^\top$ and covariance matrix $\Sigma=0.1I_2$ and
$B_\mathrm{ex,PM}=I_2$.  We define the initial set as $\mathcal{I}=[0,2]^2$, and 
$\psi_{\bx[0]}$ as a uniform distribution over $\mathcal{I}$. 
We use Algorithm~\ref{algo:FSR_thm1}
to compute $\FSRset(t, \mathcal{I})$ and $\psi_{\bx}[\cdot;t]$ for $t \in [0, 10]$ seconds. 
Figure~\ref{fig:FSReachSetPM} shows the evolution
of $\FSRset(\cdot, \mathcal{I})$ over time (plotted using MPT~\cite{MPT3}),
Figure~\ref{fig:FSRPMPM} shows the FSRPM for the system at $t=10$ and
Figure~\ref{fig:FSRCompute} compares the runtime of the
Algorithm~\ref{algo:FSR_thm1} with the dynamic programming approach, with gridding of the state-space and disturbance 
with a grid size of $0.1$ in each dimension. 


All computations in this paper were performed using MATLAB on an Intel Core i7
CPU with 2.10GHz clock rate and 8 GB RAM.

\subsection{Rigid body obstacles}
\label{sub:rigid_body}

We first extend the FSR set from point-mass obstacles (Subsection \ref{sub:nonlin}) to rigid body obstacles.

Presume that the center of mass (referred to as the center, in shorthand) of the rigid body obstacle is described by
$\bx[\cdot]$.  
The set $O(\bx[t])\subseteq \mathcal{X}$ describes 
set of states occupied by the obstacle at time $t$ when the obstacle's
center is $\bx[t]$, 
that is, $O(\bx[t]) = \{ y \: | \: h(y-\bx) \geq 0\}$ for some function $h: \mathbb R^n \rightarrow \mathbb R$ which implicitly describes the geometry of the obstacle.  For example, for a unit square obstacle, one possible geometry function is $h(z) = 1-\|z\|_{\infty}$.

We define an occupancy function of an obstacle $\phi^r_{\bx}(\bar{y};t):  \mathcal{X} \times[0,T] \rightarrow [0,1]$ to evaluate the
probability of a point $\bar{y}\in \mathcal{X}$ being covered by the rigid
body obstacle. 
\begin{align}
    \phi_{\bx}^{r}(\bar{y};t) &= \psi_{\bx}[\left\{\bar{z}\in
\FSRset(t,\mathcal{I})\vert\bar{y}\in O(\bar{z})\right\};t] \nonumber \\
&=\sum_{\bar{z}\in \FSRset(t, \mathcal{I})}
\psi_{\bx}[\bar{z};t]\textbf{1}_{O(\bar{z})}(\bar{y}). \label{eq:prob_rigid}
\end{align}
The description (\ref{eq:prob_rigid}) follows from the inclusion-exclusion principle and the observation that the states of the rigid body centers are mutually exclusive events. 

The occupancy function provides the collision probability with the rigid body obstacle. 
Note that the occupancy function is not a probability measure
since $\sum_{\bar{y}\in \mathcal{X}}\phi_{\bx}^{r}(\bar{y};t)\neq 1$. Also, the occupancy function lies in the interval $[0,1]$ since it is a sum of
nonnegative numbers, and it is upper bounded by $\sum_{\bar{z}\in \FSRset(t,
\mathcal{I})}\psi_{\bx}[\bar{z};t]=1$.

We denote the $\alpha$-superlevel set of the occupancy function as 
    \begin{align}
        S_\alpha(t;\phi^r_{\bx})&=\{\bar{y}\in \mathcal{X}\vert
    \phi^r_{\bx}(\bar{y};t)\geq \alpha\}\label{eq:sup_level_set}
    \end{align}
For $\alpha >0$,  $S_\alpha(t;\phi^r_{\bx})$ is the ``avoid'' set for obstacle avoidance problems when probabilistic safety must be assured with at least likelihood $\alpha$.  Note that $S_0(t;\phi^r_{\bx})$ is equivalent to the conservative avoid set generated by worst-case reachability formulations \cite{wu2012guaranteed,chung2006coordinated,althoff2014online,lin2015collision}.  

To utilize existing integer-programming based obstacle avoidance methods, we must have $\alpha$-superlevel sets of the occupancy function to be convex, or a union of convex sets. 

We first define 
a function $D_j(\alpha, t)$
which describes the underlying cause of an occupancy function taking a value above the threshold $\alpha$.
Given 
$\FSRset(t,\mathcal{I})$, let
\begin{align}
D_j(\alpha,t) &= \{\bar{z}_k \in \FSRset(t,\mathcal{I}) : \sum_k\psi_{\bx}[\bar{z}_k;t]>\alpha  \notag\\
& \text{and} \cap_k O(\bar{z}_k) \neq \emptyset\} \label{eq:Dj}\\
j &\in \{1,2,3...2^{|\FSRset(t,\mathcal{I})|}\} \notag
\end{align}
be a set of possible rigid body centers, whose corresponding rigid bodies create overlap with an associated probability of obstacle occupancy greater than $\alpha$  \eqref{eq:prob_rigid}. 
%
We denote $\mathscr{D}_z(\alpha,t)$ the collection of all such sets at time $t$.

To demonstrate, consider the scenario shown in Figure \ref{fig:dz}.  Overlap in possible obstacle 
positions $\bar{z}_1, \bar{z}_2$ generates a region of the state-space where likelihood of collision 
is higher than $\alpha$, even though $\psi_{\bx}[\bar{z}_1;t], \psi_{\bx}[\bar{z}_2;t] < \alpha$.  
We denote this region, as well as other regions (e.g., $\psi_{\bx}[\bar{z}_3;t]$) with likelihood higher than $\alpha$ through $D_1(\alpha, t) = \{\bar{z}_1,\bar{z}_2\}$ and $D_2(\alpha, t) = \{\bar{z}_3\}$.  Essentially, $D_j(\alpha, t)$ identifies the relevant obstacles through their centers.



\begin{prop}
    For a single rigid body $O(\cdot)$ that is convex, the $\alpha$-superlevel sets of the occupancy
    function $\phi^r_{\bx}[\bar{y};t]$ (\ref{eq:prob_rigid}) is a union of convex
    sets.\label{prop:cvx_union_obs}
\end{prop}

\begin{proof}
    Define regions of overlap described by (\ref{eq:Dj}) for a given likelihood $\alpha$ and FSRPM $\psi_{\bx}[\bar{y};t]$. Then, 
    \begin{align}
        S_\alpha(t;\phi_{\bx}^r)&=\bigcup_{D_j(\alpha,t)\in
    \mathscr{D}_z(\alpha,t)}\bigcap_{\bar{z}\in
    D_j(\alpha,t)}O(\bar{z})\subseteq \mathcal{X} \label{eq:sup_lvl_single}
    \end{align}
    for $\bar{z}\in \mathcal{X}$.
    The proof is complete with the observation that intersection preserves
    convexity.
\end{proof}
Hence Proposition 1 solves Problem \ref{prob:prob}.
\begin{figure}[h]
 \includegraphics[width=0.7 \linewidth]{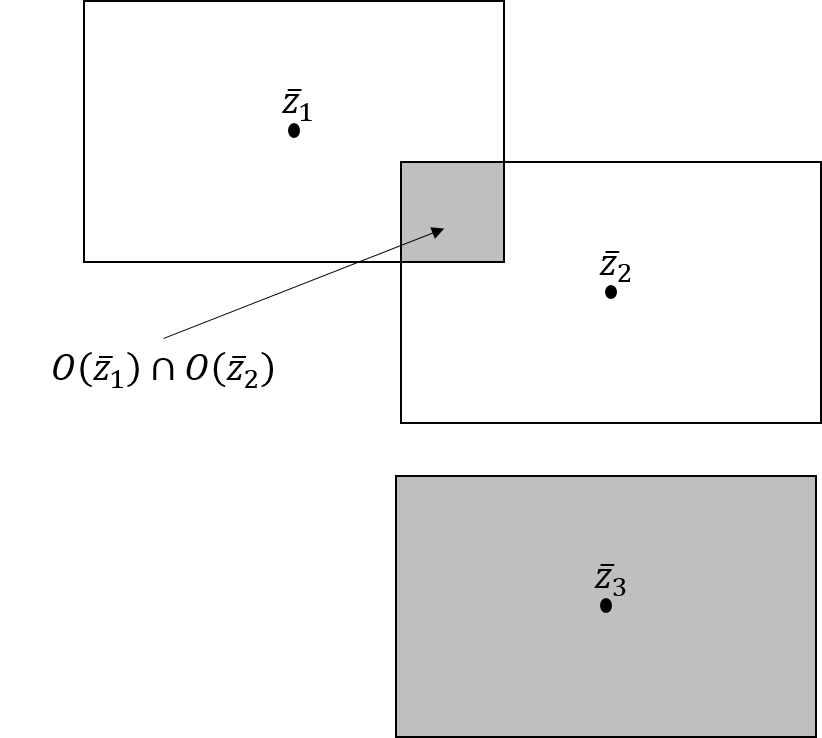}
\caption{The shaded region has a likelihood of collision greater than $\alpha$. Consider a rigid-body obstacle, with possible obstacle locations $\bar z_1, \bar z_2, \bar z_3$ at time $t$.  Presume that $\psi_{\bx}[\bar{z}_1;t] < \alpha$, $\psi_{\bx}[\bar{z}_2;t]< \alpha$, $\psi_{\bx}[\bar{z}_3;t] \geq \alpha$.  Note that overlap between obstacles with centers $\bar{z}_1$ and $\bar{z}_2$ creates a likelihood greater than $\alpha$ for $O(\bar{z}_1) \cap O(\bar{z}_2)$, so that $\mathscr{D}_z(\alpha,t) = \{ D_1(\alpha, t), D_2(\alpha, t)\}$ with $D_1(\alpha, t) = \{z_1, z_2\}$ and $D_2(\alpha, t) = \{z_3\}$. Thus, $S_\alpha(t;\phi_{\bx}^r)=(O(\bar{z}_1)\cap O(\bar{z}_2))\cup O(\bar{z}_3)$.}
\label{fig:dz}
\end{figure}


Since the indicator function for the obstacle geometry in (\ref{eq:prob_rigid}) can be equivalently 
expressed as $\textbf{1}_{O(\bar{z})}(\bar{y}) =
\textbf{1}_{O(0)}(\bar{y} - \bar{z})$, 
the occupancy function \eqref{eq:prob_rigid} can be re-written as
\begin{align}
\phi_{\bx}^{r}(\bar{y};t) &= \sum_{\bar{z}\in \FSRset(t,\mathcal{I})}
    \psi_{\bx}[\bar{z};t] \textbf{1}_{O(0)}(\bar{y} - \bar{z}) \nonumber \\
&= \left(\psi_{\bx}[\cdot;t] \ast
\textbf{1}_{O(0)}(\cdot)\right)(\bar{y}) \label{eq:conv_ind}
\end{align}
Equation (\ref{eq:conv_ind}) is similar to the concept of blurring in
image processing, in which an image (in our case, $\psi_{\bx}[\cdot;t]$), is
convoluted with a shift-invariant point spread function, (in our case,
$\textbf{1}_{O(0)}(\cdot)$). 
Such a formulation enables potential use of tools from
image processing for the computation of $\phi_{\bx}^{r}[\cdot;t]$ and
its support for rigid body obstacles.

Now, we analyze the convexity property of (\ref{eq:sup_level_set}) for multiple moving obstacles. For $N_\mathrm{obs}$ homogeneous obstacles, we denote the concatenated random
vector of obstacle centers as $\bX=[\bx_1\ \bx_2\ \ldots\
\bx_{N_\mathrm{obs}}]\in \mathcal{X}^{N_\mathrm{obs}}$. 
We presume that the obstacles do not interact with each other, and hence are stochastically
independent. For a given obstacle characterization 
$Z=[\bar{z}_1\ \bar{z}_2\ \ldots\ \bar{z}_{N_\mathrm{obs}}]\in \mathcal{X}^{N_\mathrm{obs}}$, the forward
reach set and the probability measure of the obstacle configuration are described by
\begin{align}
    \FSRset_{\bX}(t,
    \mathcal{I})&=\bigtimes_{i=1}^{N_\mathrm{obs}}\FSRset_{\bx_i}(t,
    \mathcal{I})\label{eq:FSRset_config} \\
    \psi_{\bX}[Z;t]&=\prod_{i=1}^{N_\mathrm{obs}}\psi_{\bx_i}[\bar{z}_i;t].\label{eq:FSRPM_config}
\end{align}
Computation of \eqref{eq:FSRset_config}, \eqref{eq:FSRPM_config} relies on
Algorithm~\ref{algo:FSR_thm1} to compute  
\eqref{eq:FSRPM_recursive_sum}, \eqref{eq:FSRPM_psif}
for each obstacle individually. 

We then define the joint occupancy function $\phi_{\bX}^r(\bar{y},t):  \mathcal{X}\times[0,T] \rightarrow [0,1]$ for a group of obstacles as the probability of any obstacle in the group occupying a state $\bar{y}$.  Because of the mutual exclusivity of the configurations, the joint occupancy function is described by
\begin{align}
    \phi_{\bX}^{r}(\bar{y};t) &= \sum_{Z\in \FSRset_{\bX}(t, \mathcal{I})}
    \psi_{\bX}[Z;t]\textbf{1}_{\overline{O}(Z)}(\bar{y})
    \label{eq:prob_rigid_config_1}
\end{align}
with $\overline{O}(Z)=\cup_{i=1}^{N_\mathrm{obs}}O(\bar{z}_i)$.  
Similarly to (\ref{eq:Dj}), we define $\bar{D}_j(\alpha,t)$ as the sets of configurations, $j\in \{1,2,3,...,2^{|\FSRset_{\bX}(t,\mathcal{I})|}\}$, whose probability of occurrence is greater than $\alpha$ and the resulting overlap is non-empty, and define $\mathscr{D}_Z(\alpha,t) $ as the collection of such sets for a given time $t$ in the configuration space $Z$. 

Using an approach similar to that of
Proposition~\ref{prop:cvx_union_obs}, we can show that the $\alpha$ superlevel set of
$\phi_{\bX}^{r}(\bar{y};t)$ is
\begin{align}
    S_\alpha(t;\phi_{\bX}^{r})&=\bigcup_{\bar{D}_j(\alpha,t)\in
\mathscr{D}_Z(\alpha,t)}\bigcap_{Z\in \bar{D}_j(\alpha,t)} \bigcup_{i=1}^{N_\mathrm{obs}}O(\bar{z}_i). \label{eq:sup_config}
\end{align}
Note that for $N_\mathrm{obs}=1$, \eqref{eq:sup_config}
reduces to \eqref{eq:sup_lvl_single}, as expected. 
We see from (\ref{eq:sup_config}) that the avoid set for multiple moving obstacles is in general non-convex, and cannot be expressed as a union of convex avoid sets. Thus, to utilize integer programming based methods, the sets $S_\alpha(t;\phi_{\bX}^{r})$ at every $t$ must be overapproximated as a union of convex sets.  Although this is typically computationally expensive, we provide one such method in the next Section.  

An alternative interpretation of \eqref{eq:prob_rigid_config_1} 
can be given by using events $\mathcal{E}_i$, which occur when
$\mathbf{1}_{O(\bx^i)}(\bar y) = 1$.
Essentially, the event $\mathcal E_i$ corresponds to the $i^\mathrm{th}$
obstacle occupying the state $\bar{y}\in \mathcal{X}$.  Note that the event $
\mathcal{E}_i$ depends only on the state of $i^\mathrm{th}$ obstacle center, and
does not provide any restrictions on the centers of other obstacles in the
configuration. Equation \eqref{eq:prob_rigid_config_1} can be rewritten as 
\begin{align}
    \phi_{\bX}^{r}(\bar{y};t) &= \Exp_{\bX}^t\left[\textbf{1}_{\overline{O}(Z)}(\bar{y}) \right]= \Prob_{\bX}^t\left\{  \bigcup_{i=1}^{N_\mathrm{obs}}
    \mathcal{E}_i\right\}
    \label{eq:prob_rigid_config_2}
\end{align}
where $\Prob_{\bX}^t$ denotes the joint probability measure associated with the
configuration of the obstacles. 
Such a formulation is important for constructing 
an overapproximation of avoid set (and hence under-approximation of the collision-free set)
that can be represented as the union of convex sets.
%
%

We define ``safety'' as ensuring that the probability of collision
of the robot with any of the obstacles at any given time $t\in[0,T]$ is less
than a specified threshold, $\alpha\in[0,1]$. We define the safe set $\mathrm{SafeSet}[t;\alpha]$ as the complement of the set $ S_\alpha(t;\phi_{\bX}^{r})$, such that
\begin{align}
    \mathrm{SafeSet}[t;\alpha]&=\left\{\bar{y}\in \mathcal{X}\middle\vert
\phi^r_{\bX}(\bar{y})< \alpha \right\}= \mathcal{X}\setminus S_\alpha(t;\phi_{\bX}^r).\label{eq:safety}
\end{align}

Note that (\ref{eq:safety}) is not
guaranteed to be convex even if the obstacles are convex.
Convexification methods \cite{accikmecse2013lossless} for evaluation of \eqref{eq:sup_config} and \eqref{eq:safety} would need to be implemented online, and are computationally expensive.  
Hence, for computational tractability, at each time step $t$, 
we underapproximate 
\eqref{eq:safety} using the definition of $\phi_{\bX}^r$ in
\eqref{eq:prob_rigid_config_2},
\begin{align}
    \underline{\mathrm{SafeSet}}[t;\alpha]&=\bigcap_{i=1}^{N_\mathrm{obs}}\left\{\bar{y}\in
    \mathcal{X} \middle\vert \Prob_{\bX}^t(\mathcal{E}_i)<
\frac{\alpha}{N_\mathrm{obs}} \right\}.
\label{eq:safety_P_approx}
\end{align}
Since $\Prob(\cup_i
\mathcal{E}_i)\leq \sum_i\Prob( \mathcal{E}_i)$, 
\begin{align}
    \underline{\mathrm{SafeSet}}[t;\alpha]\subseteq
    \mathrm{SafeSet}[t;\alpha].\label{eq:SafeSet_underapprox}
\end{align}
By construction, $\mathcal{E}_i$ restricts the state of the $i^\mathrm{th}$
obstacle alone. Therefore, \eqref{eq:safety_P_approx} can be computed using
\eqref{eq:sup_lvl_single} as
\begin{align}
    \underline{\mathrm{SafeSet}}[t;\alpha]&=\bigcap_{i=1}^{N_\mathrm{obs}}\left( \mathcal{X}\setminus
S^i_\frac{\alpha}{N_\mathrm{obs}}(t;\phi_{\bx}^r)\right)
.\label{eq:safety_P_approx_2} \\
 &=\mathcal{X}\setminus \bigcup_{i=1}^{N_\mathrm{obs}}\left( S^i_\frac{\alpha}{N_\mathrm{obs}}(t;\phi_{\bx}^r)\right)\label{eq:safety_P_approx_3}
\end{align}
Hence Problem \ref{prob:prob3} is solved, since 
 $\cup_{i=1}^{N_{obs}}S^i_\frac{\alpha}{N_\mathrm{obs}}(t;\phi_{\bx}^r)$ in (\ref{eq:safety_P_approx_3}) is an overapproximation that can be written as a union of convex sets.
Here, $S^i_\frac{\alpha}{N_\mathrm{obs}}$ is the $\left(\frac{\alpha}{N_\mathrm{obs}}\right)$-superlevel set of the occupancy function of the $i^{th}$ obstacle at time $t$. Note that 
$S^i_\frac{\alpha}{N_\mathrm{obs}}(\cdot,t)$ for every $t=[0,1, \cdots ,T]$ and for every obstacle $i$ can be computed offline, and hence any form of online convexification is avoided with this overapproximation.

The particular under-approximation will be specific to the problem at hand.  In general, integer variables must be introduced to accommodate each obstacle, and the sets in (\ref{eq:sup_config}) approximated via a set of linear constraints.  
We demonstrate this approach in the next Section.


\section{Application to obstacle avoidance}
\label{sec:motionPlanning}


We now consider the specific problem of 
robot navigation in an environment with ${N_\mathrm{obs}}$ rigid body obstacles
moving in straight lines with stochastic velocities. 
We use integer programming \cite{schouwenaars2002safe,mellinger2012mixed} in a receding horizon control framework to drive
the robot to the desired goal
$\bar{x}_G\in \mathcal{X}$ in finite time,
while ensuring a probabilistic guarantee of safety.  
We presume that robot and obstacle positions are known at each instant.  

We model the robot as a point mass under state-feedback control
\begin{align}
    \bar{x}_R[t+1]&=\bar{x}_R[t]+B_R u[t] \label{eq:robot_dyn_SIM1}
\end{align} 
with state $\bar{x}_R[t]\in \mathcal{X}=
\mathbb{R}^{2\times 1}$ that represents robot position and input $u[t]\in \mathcal{U}\subseteq \mathbb{R}^2$.  
The input matrix is $B_R=T_sI_2$, with sampling time $T_s$.  

The obstacles have identical dynamics and do not interact with each other, and have rigid bodies 
that are unit boxes with fixed heading. 
In the absence of any rotation, the obstacle position is completely
characterized by the dynamics of the center. The dynamics of the center of the
$k^\mathrm{th}$ obstacle is described
\begin{align}
    \bx^k_o[t+1] &= \bx^k_o[t] + B_o \bw^k[t]
\end{align} 
with state $\bx^k_o[t]\in \mathcal{X}$,
stochastic velocity $\bw^k[t]\in \mathcal{W}_{o,d}^2$ described by an i.i.d. process, 
and disturbance matrix $B_o = B_R$.  The disturbance set $ \mathcal{W}_{o,d}\subseteq \mathbb{R}$ describes possible obstacle velocities. 
We define the probability mass function of the velocity vector
$\bw^k$ to be $\psi_{\bw^k}[z]$, hence
the state $\bx_o^k[t]$ is a random
vector in the probability space $( \mathcal{X}, \sigma( \mathcal{X}),
\Prob_{\bx_o}^{t,\bar{x}_o,k})$ for a given initial position $\bar{x}_o\in
\mathcal{X}$. The probability measure associated with $k^\mathrm{th}$ obstacle
$\Prob_{\bx_o}^{t,\bar{x}_o,k}$ is induced from the product measure associated
with $\psi_{\bw^k}$ and depends on the initial position $\bar{x}_0$ and time
$t$. 


\begin{figure*}[h]
    \centering 
    \begin{align}
       \mbox{Prob A:} 
             &\begin{array}{rcl}
            \underset{\pi}{\operatorname{minimize}} & & J(\pi;\bar{x}_R[0],\bX[\cdot])=\sum_{t=0}^T \left\{(\bar{x}_R[t]-\bar{x}_G)^\top
        Q(\bar{x}_R[t]-\bar{x}_G)+\pi(t,\bar{x}_R[t],\bX[t])^TR^u(t,\bar{x}_R[t],\bX[t])\right\}\\
          \mbox{subject to}& & \left\{\begin{array}{rll}
            \bar{x}_R[t] &=\bar{x}_R[t-1]+B_R\pi(t,\bar{x}_R[t-1],\bX[t])
        &t=1,\ldots,T, \\
            \bar{x}_R[t] &\in \mathrm{SafeSet}[t;\alpha] &t=1,\ldots,T  \\
            \pi(\cdot) &\in \mathcal{M}                  \\
        \end{array}\right.
      \end{array}\nonumber 
    \end{align}
    \rule{\textwidth}{0.5pt}
\end{figure*}


We wish to solve Problem \probref{Prob A}.
The control policy $\pi(t,\bar{x}_R[t],\bX[t]):[0,T-1]\times \mathcal{X}^{1+N_\mathrm{obs}} \rightarrow
\mathcal{U}$ is a state-feedback control with the set of feasible policies
$\pi(\cdot)$ denoted by $ \mathcal{M}$. Here, $Q$ and $R^{u}$ are symmetric positive definite
matrices of appropriate dimensions. 

A conservative solution to Problem \probref{Prob A} can be found by solving the
following optimization problem:
\begin{align}
    \begin{array}{rl}
        \mbox{Prob B:} &\begin{array}{rl}
      \mbox{minimize} & J(\pi;\bar{x}_R[0],\bX[\cdot]) \\
          \mbox{subject to} & \left\{\begin{array}{rll}
         \bar{x}_R[t]&\mbox{ by \eqref{eq:robot_dyn_SIM1} with }\pi  &\forall t\\
        \bar{x}_R[t]&\in \underline{\mathrm{SafeSet}}[t;\alpha] &\forall t  \\
      \pi&\in \mathcal{M} & \\
        \end{array}\right.
  \end{array}\nonumber 
  \end{array}
\end{align}

We replace the constraint
$x_R[t]\in \underline{\mathrm{SafeSet}}[t;\alpha]$ in
Problem~\probref{Prob B} by defining 
\begin{equation}
K_i=\{\bar{y}\in \mathcal{X}\vert P_{i}\bar{y}\leq \bar{q}_i, P_i \in
\mathbb{R}^{n_i\times 2},\bar{q}_i \in \mathbb{R}^{n_i}\}
\end{equation}
such that 
$\mathcal{X}\setminus(\cup_{i=1}^{N_s}
K_i)\subseteq\underline{\mathrm{SafeSet}}[t;\alpha]$, resulting in 
%
%
the following constraint set for
$i=1,\ldots,N_s$:
\begin{subequations}
\label{eq:const}
\begin{align}
    \delta_{i,l} &\in\{0,1\}, & l=1,\ldots,n_i \label{eq:cons1a}\\
    -\bar{p}_{i,l}[t]^\top\bar{x}_R[t] &< -q_{i,l}[t] + M_\mathrm{big}\delta_{i,l}\label{eq:cons1b}\\
    \sum_{l=1}^{n_i}\delta_{i,l} &\leq (n_i-1) \label{eq:cons1c}
\end{align}
\end{subequations}
Here, $\bar{p}_{i,l}[t]$ and $q_{i,l}[t]$ are the $l^{th}$ row of matrix
$P_i[t]$ and $l^{th}$ element of vector $\bar{q}_i[t]$ respectively. The term
$M_\mathrm{big}$ is a large number that facilitates the constraint satisfaction.
The constraint \eqref{eq:cons1c} ensures that at least one of the binary
variables $\delta_{i,l}=0$ for every $i$. This formulation ensures the robot avoids every avoid set $i=1,2, \cdots, N_s$.

We implement the problem with the following parameters: $T_s=0.2$,
$T=50$, the stochastic speed set $\mathcal{W}_{o,d}=\{3,2.5,1.5,2,1,0.8,0.5,0.1\}$ m/s with probabilities $\psi_{\bw^k}[z]\in\{0.05,0.05,0.30,0.20,0.25,0.10,0.04,0.01\}$. 

The input space for the robot is $
\mathcal{U}=[-0.2,1]\times[0.1,1]$, so that it cannot stop in the $y$-direction.  
Note that the average velocity of each obstacle is
$1.476$ m/s while the maximum robot velocity in both directions is $1$ m/s,
which is about two-thirds 
the obstacle's maximum velocity. The robot is disadvantaged because it is slower than the obstacles.

To compute the forward stochastic reach sets and occupancy function, we discretize the state space with a resolution of $0.05$ and follow Algorithm \ref{algo:FSR_thm1}. 
We use YALMIP \cite{lofberg2004yalmip} with the Gurobi \cite{gurobi} solver to solve Problem \probref{Prob B} with the constraint in (\ref{eq:const}).  The computation took approximately $0.04$ seconds to complete.

Results are shown in Figure \ref{fig:sim1} from a single initial obstacle-robot configuration.  We compare our probabilistic approach with the case in which $\alpha = 0$, which is equivalent to the result from the conservative min-max
solution 
in~\cite{mitchell2005time,wu2012guaranteed, chung2006coordinated, althoff2014online, lin2015collision}.  
Note that the min-max solution becomes infeasible at approximately $1.8$ seconds (10 time steps).  With $\alpha = 0.045$, meaning that obstacles should be avoided with likelihood of 0.95, feasible solutions are found for the entire time horizon.

\begin{figure}[h]
 \includegraphics[width=0.65\linewidth]{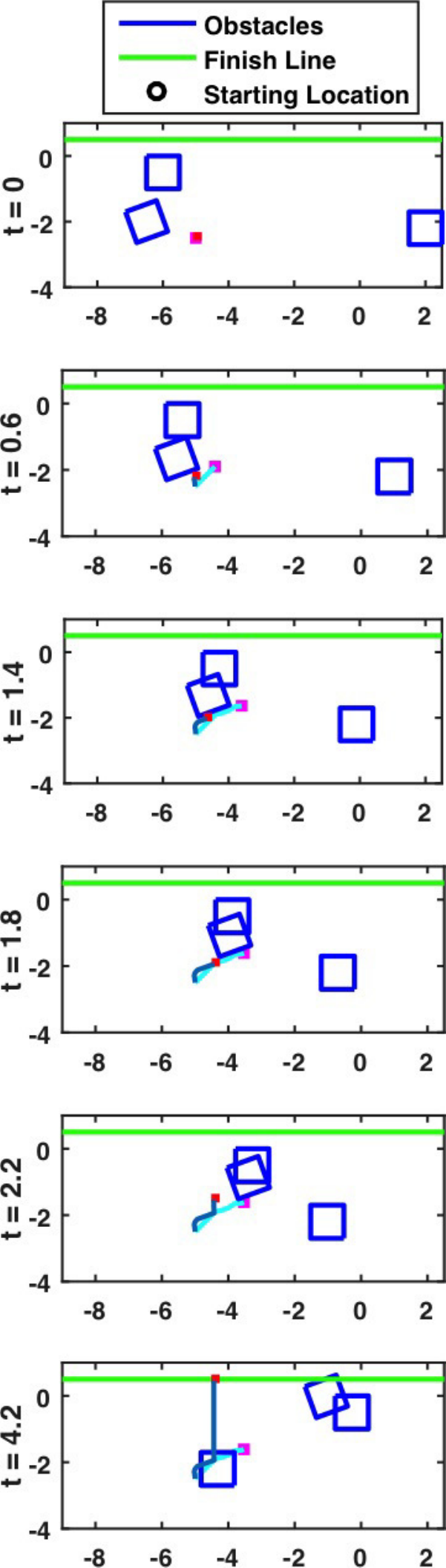}
 \caption{Snapshots of stochastically moving obstacles and robots and their trajectories. The red robot's trajectory is indicated by the blue line, and the pink robot's trajectory is indicated by the cyan line. The red robot uses $\alpha = 0.045$ while the pink robot uses $\alpha =0$ (the min-max problem).}
\label{fig:sim1}
\end{figure}


\section{Conclusions and Future work}
\label{sec:conc}

This paper provides a method for computing the forward stochastic reachable set and probability measure, with application to obstacle avoidance. The
method handles uncontrolled nonlinear systems, or systems with a known controller,
as well as an affine disturbance that captures the stochastic element. 
We have described how the forward stochastic reachable set and probability measure can be used to generate 
an occupancy constraint that can be written as a union as convex sets, and hence is amenable to use in 
existing integer programming based methods for collision avoidance over a finite horizon.

Future work includes the extension to problems with an uncountable sample space, and development of computationally efficient online control methods. We also anticipate application of these techniques to (dynamic) target reaching problems.



\bibliographystyle{ieeetr}
\bibliography{obstacleavoid_acc2017}

\end{document}